\documentclass[twocolumn, aps, pra, 10pt]{revtex4-2}
\usepackage{graphicx} 
\usepackage[utf8]{inputenc}
\usepackage{amsmath, amssymb, amsfonts, amsthm}
\usepackage{mathtools}
\usepackage{physics}
\usepackage{microtype}
\usepackage{xcolor, soul}
\usepackage{hyperref}

\newtheorem{theorem}{Theorem}

\newtheorem{corollary}[theorem]{Corollary}

\newcommand{\BargmannSet}[1]{\mathcal{B}_{#1}}
\newcommand{\RootSet}[1]{\mathcal{R}_{#1}}
\newcommand{\ceil}[1]{{\lceil#1\rceil}}

\begin{document}

\title{An Elementary Characterization of Bargmann Invariants}
\author{Sagar Silva Pratapsi}
\email[Corresponding author: ]{spratapsi@uc.pt}
\affiliation{CFisUC, Department of Physics, University of Coimbra, P-3004-516 Coimbra,
Portugal}
\author{João Gouveia}
\affiliation{CMUC, Departament of Mathematics, University of Coimbra, P-3001-454 Coimbra, Portugal}
\author{Leonardo Novo}
\affiliation{International Iberian Nanotechnology Laboratory (INL)
 Av. Mestre José Veiga s/n, 4715-330 Braga, Portugal
}
\author{Ernesto F. Galv\~ao}
\email[Corresponding author: ]{ernestogalvao@id.uff.br}
\affiliation{International Iberian Nanotechnology Laboratory (INL)
 Av. Mestre José Veiga s/n, 4715-330 Braga, Portugal
}

\affiliation{Instituto de F\'isica, Universidade Federal Fluminense, Av. Gal. Milton Tavares de Souza s/n, Niter\'oi, RJ, 24210-340, Brazil}
\date{May 2025}

\date{\today}

\begin{abstract}
    Bargmann invariants, also known as multivariate traces of quantum states $\Tr(\rho_1 \rho_2 \cdots \rho_n)$, are unitary invariant quantities used to characterize weak values, Kirkwood-Dirac quasiprobabilities, out-of-time-order correlators (OTOCs), and geometric phases. 
    Here we give a complete characterization of the set $B_n$ of complex values that $n$-th order invariants can take, resolving some recently proposed conjectures. We show that $B_n$ is equal to the range
    of invariants arising from pure states described by Gram matrices of circulant form. We show that both ranges are equal to the $n$-th power of the complex unit $n$-gon, and are therefore convex, which provides a simple geometric intuition. Finally, we show that any Bargmann invariant of order $n$ is realizable using either qubit states, or circulant qutrit states.
\end{abstract}

\maketitle

\section{Introduction}

The characterization of quantum states and their relationships is fundamental to quantum mechanics and its applications. Among the various tools available for this purpose, 
Bargmann invariants \cite{Bargmann64, simon1993Bargmann} have emerged as a fundamental tool to characterize such relationships.
The $n$-th order Bargmann invariants are also known as multivariate traces, due to their form $\Tr(\rho_1 \rho_2 \cdots \rho_n)$, where $\rho_i$ are quantum states in some Hilbert space. They are manifestly invariant under the application of the same unitary on all states, which makes them ideal for describing basis-independent concepts in quantum theory, such as basis-independent coherence \cite{galvao2020quantumandclassical, Desi21} and contextuality \cite{wagner2024inequalities}. They arise naturally in the study of Kirkwood-Dirac quasiprobabilities \cite{arvidssonshukur2024properties} and out-of-time-order correlators (OTOCs) \cite{wagner2024quantum}, crucial for the characterization of quantum chaos and information scrambling. They describe weak values \cite{wagner2024quantum, wagner2023simple} and geometric phases \cite{simon1993Bargmann}, and their measurement can be used to certify the need for complex-valued amplitudes in quantum theory \cite{oszmaniec_measuring_2024, fernandes2024unitaryinvariant}, a resource known as imaginarity \cite{Wu21, Renou21}. Bargmann invariants can be measured using simple Hadamard test circuits \cite{oszmaniec_measuring_2024}; alternatively, using constant quantum depth adaptive circuits \cite{quek24}, or simpler static circuits that make use of available classical information on some of the states \cite{simonov25}.

Despite their importance, a complete characterization of the possible complex values that Bargmann invariants can take has remained an open problem. Previous work has established partial results, in particular a complete characterization of the range for invariants of order 3 \cite{fernandes2024unitaryinvariant} and 4 \cite{Zhang25}. A characterization was given also for the particular case of tuples of states whose Gram matrix is of circulant form \cite{Li25}. The range of complex values that can be taken by invariants of order 5 has remained unknown up to now. Fifth-order invariants of the form $\text{Tr}(\rho_1\rho_2\rho_3\rho_4\rho_5)$ are used to describe observable incompatibility \cite{gao2023measuring} and OTOCs~\cite{halpern2018quasiprobability,wagner2024quantum}. Invariants of arbitrarily high order describe sequential weak measurements \cite{mitchison2007sequential}, extended Kirkwood-Dirac quasiprobability representations \cite{halpern2018quasiprobability}, and recently proposed higher-order OTOCs \cite{abanin25}. Importantly, they also describe geometric phases \cite{simon1993Bargmann,fernandes2024unitaryinvariant}. Already in the early works by Pancharatnam \cite{Pancharatnam56} and Berry \cite{Berry09}, it was noted that the geometric phase, measurable in multiple physical platforms, is the phase of a Bargmann invariant associated with the chosen closed particle trajectory (see \cite{Chruscinski04} for a review). Learning about the range of values that can be taken by Bargmann invariants means we can make novel quantitative predictions about general geometric phase experiments. Besides the incomplete foundational understanding of the varied phenomena we have mentioned, our lack of understanding of the range of Bargmann invariants limits practical applications involving measurements of these invariants. For example, these can be used to quantify nonclassicality in quantum information processing and metrology  \cite{wagner2024quantum, fernandes2024unitaryinvariant, oszmaniec_measuring_2024}, where it is usually associated with negative or complex values. In photonic quantum computation, the characterization of multiphoton indistinguishability relies on the characterization of the joint state of the internal photonic degrees of freedom using measurements of invariants, known in this field as collective photonic phases \cite{shchesnovich2015partial, Menssen2017, Jones2020, oszmaniec_measuring_2024}.

In this work, we provide a complete characterization of the sets of complex values that $n$th-order Bargmann invariants can take. We show that any Bargmann invariant can be realized as a Bargmann invariant of a tuple of pure qutrit states whose Gram matrix has a circulant form, establishing the equivalence between the general and qutrit-circulant cases. Furthermore, we show these sets have a simple geometric structure—they are convex and correspond to powers of regular polygons in the complex plane. Finally, we prove that any Bargmann invariant is realizable using only qubit states, providing a practical foundation for experimental tests of applications.

\subsection{Definitions}
Given an ordered tuple $\rho$ of density matrices $\rho_1, \ldots, \rho_n \in  \mathcal{D}(\mathcal H_d)^n$, where $\mathcal H_d$ is a $d$-dimensional Hilbert space, its Bargmann invariant is the cyclic trace of their product,
\begin{equation}
    \Delta_\rho =
    \Tr \{
        \rho_1 \, \cdots\,  \rho_n
    \}.
\end{equation}
When the states are pure, that is, when $\rho_i = \ketbra{v_i}{v_i}$ for an ordered collection $V$ of unit vectors $\ket{v_1}, \ldots, \ket{v_n} \in \mathcal{H}_d^n$, it is clear that
\begin{equation}
    \Delta_\rho
    =
    \Delta_V
    :=
    \braket{v_1}{v_2} \, \cdots \, \braket{v_n}{v_1}.
\end{equation}
When the ordered tuple $V$ is clear from context, we may write the invariants using the indices of the vectors directly, e.g.\ $\Delta_V = \Delta_{1\,2\,\cdots\,n}$.
This notation is useful when taking the Bargmann invariants in different orders, such as $\Delta_{n\,\cdots\,2\,1} = \braket{v_n}{v_{n-1}} \, \cdots \, \braket{v_2}{v_1}.$
As a consequence of the definition, Bargmann invariants are unchanged under simultaneous transformations of the quantum states by the same unitary operator, or by gauge choices of phase for the wavefunctions.

Our goal is to characterize the sets of complex values that Bargmann invariants can take. We begin by defining $B_n$ as the set of complex values taken by Bargmann invariants of $n$ states, or order $n$, for the general case of mixed states.
It is clear that any Bargmann invariant is a convex combination of Bargmann invariants of pure states, using the spectral decomposition of density matrices.
Since the sets of Bargmann invariants of pure states turns out to be convex, it will suffice to restrict our attention to such states.

Let us then define $\BargmannSet{n, d}$ as the set of all Bargmann invariants of order $n$ and arising from a tuple of pure states of dimension $d$; and $\BargmannSet{n}$ the set of all Bargmann invariants of order $n$, assuming pure states, and regardless of dimension.
If a complex number $z$ belongs to $\BargmannSet{n, d}$, we say that $z$ is \emph{realizable} as a Bargmann invariant of order $n$ and dimension $d$.

To characterize the sets of Bargmann invariants, it is useful to work with the \emph{$n$-th roots} of Bargmann invariants, or \emph{Bargmann roots}, of order $n$,
\(
    \sqrt[n]{
        \braket{v_1}{v_2}
        \cdots
        \braket{v_n}{v_1}
    }
\).
These geometric means of inner products are readily comparable with arithmetic means, a fact that will be useful later.
Therefore, we define the sets
\begin{align}
    \RootSet{n}
    &:=
    \{z\in\mathbb{C}
    \mid
    z^n \in \BargmannSet{n}
    \},
    \\
    \RootSet{n,d}
    &:=
    \{z\in\mathbb{C}
    \mid
    z^n \in \BargmannSet{n,d}
    \}.
\end{align}
Notice that the $n$-th Bargmann roots, as the $n$-th roots of any complex number, are unique up to a multiplicative factor $w_n^k$, where $w_n := e^{i\,2\pi / n}$ is a primitive root of unity, and $0 \leq k <  n$. In other words, the sets $\RootSet{n}$, and $\RootSet{n, d}$ are closed under multiplication by $\omega_n$.

For the study of Bargmann invariants, it is useful to define the Gram matrix $G_V$ with entries $(G_V)_{ij} = \braket{v_i}{v_j},$
where $V \in \mathcal H^n$ is a tuple of $n$ vectors.
The matrix $G_V$ is positive semi-definite, and its rank is equal to the dimension of the vector subspace spanned by the tuple $V$.
If the states are normalized, the corresponding Gram matrix has unit diagonal entries and is said to be normalized.
A complex number $z$ is a Bargmann invariant if and only if
$z=G_{12} G_{23} \cdots G_{n1}$
for some positive semi-definite matrix $G$ with ones on the diagonal~\cite{chefles04}.

In this work, we will deal with tuples of states that give rise to circulant Gram matrices.
An $n\times n$ matrix $G$ is circulant if its rows are obtained from the first one by successive application of the cycle permutation $\sigma = (1,2,\ldots,n)$ i.e.,
\begin{align*}
G &=
g_0 C_n^0 + g_1 C_n^1 + \ldots + g_{n-1} C_n^{n-1} \\
&=
\begin{pmatrix}
    g_0 & g_1 & g_2  & \cdots & g_{n-1} \\
    g_{n-1} & g_0  & g_1 & \cdots &g_{n-2} \\
    & & \vdots & & \\
    g_1 & g_2 & g_3 & \cdots & g_0
\end{pmatrix},
\end{align*}
where $C_n$ is the permutation matrix associated with the permutation $\sigma$~\cite[Section 0.9.6]{HornJohnson}.

If a tuple $V$ of normalized states is such that $G_V$ is a circulant Gram matrix, we say that $V$ is a \emph{circulant tuple of states}.
We will call $\Delta_V$ a \emph{circulant Bargmann invariant},
and $\Delta_V^{1/n}$ a \emph{circulant Bargmann root}.
The set of circulant Bargmann invariants of order $n$ will be called $\BargmannSet{n \mid \mathrm{circ}}$, while the set of \emph{circulant Bargmann roots} of order $n$ will be called $\RootSet{n\mid\mathrm{circ}}$.
We define
 $\BargmannSet{n,d \mid \mathrm{circ}}$
 and 
$\RootSet{n,d\mid\mathrm{circ}}$
similarly, for some dimension $d$.

Finally, for any natural number $n$, the set $\mathcal P_n$ is the regular $n$-gon in the complex plane, i.e., the convex hull of the $n$-th roots of unity,
\begin{equation}
    \mathcal{P}_n
    :=
    \mathrm{conv}\{
        1, 
        \omega_n,
        \dots,
        \omega_n^{n-1}
    \},
\end{equation}
where again $\omega_n = e^{2i\pi/n}$ and $\mathrm{conv}(A)$ stands for the convex hull of the set $A$, which can be defined as the smallest convex set containing $A$.

\begin{figure}[th]
    \centering
    \includegraphics[width=1\linewidth]{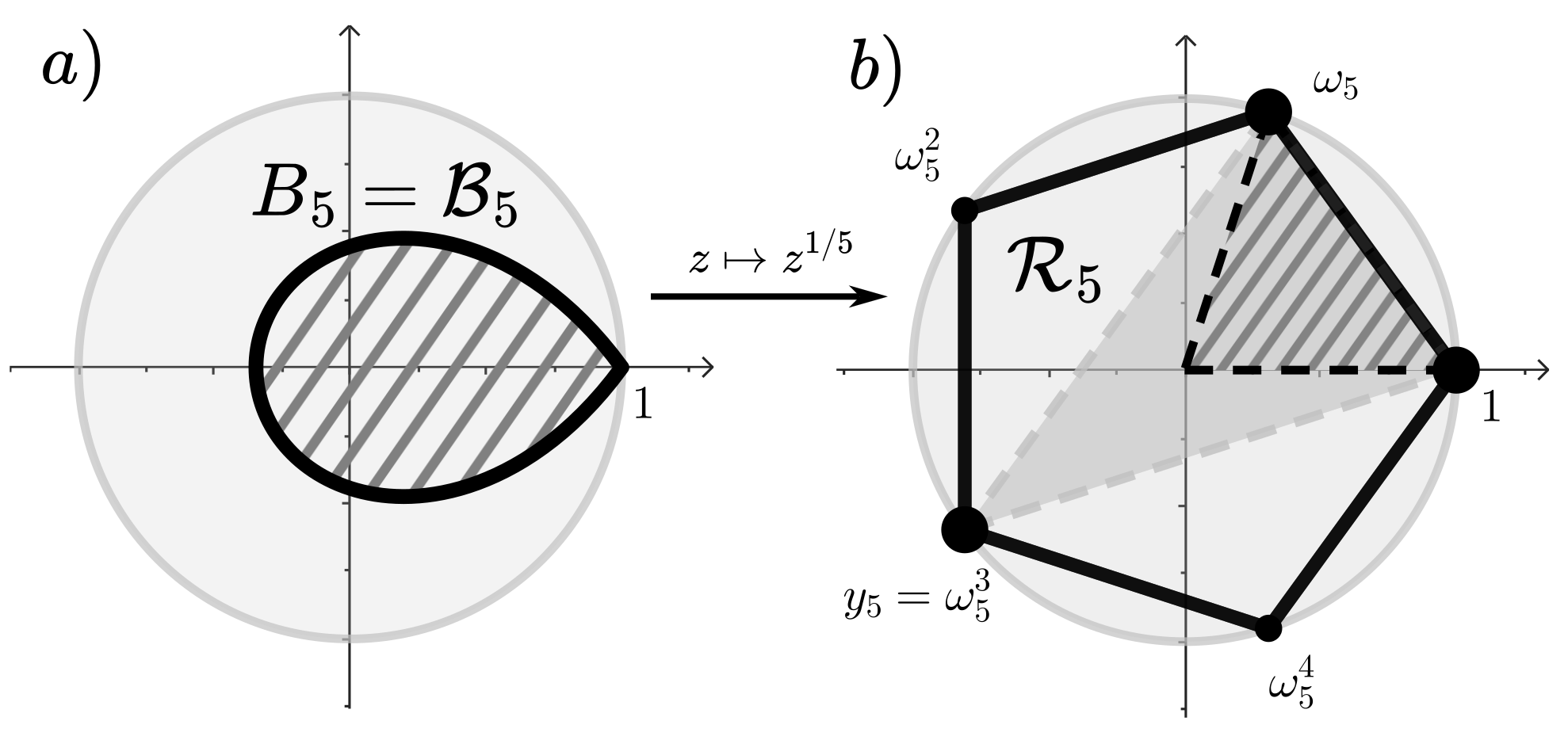}
    \caption{
    Characterization of $n$-th order Bargmann invariants, illustrated here for $n=5$, with the complex unit circle as reference.
    \textbf{a)}
    The range of values that pure Bargmann invariants of order $n$ can take, $B_{n}$, is a teardrop-shaped convex region in the complex plane (line-shaded region). 
    \textbf{b)}
    The values that the $n$-th complex roots of Bargmann invariants of order $n$ can take, $\RootSet{n}$, are the filled unit $n$-gon, $\mathcal{P}_n$.
    We prove this first for circulant tuples of states.
    The vertices, edges and interior of $\mathcal{P}_n$ are roots of Bargmann invariants realizable by states with dimensions 1, 2 (qubits) and 3 (qutrits), respectively, with circulant Gram matrices.
    The principal roots of Bargmann invariants (line-shaded triangle) are always interior to the three vertices $1$, $\omega_n$ and $y_n = \omega_n^\ceil{n/2}$, which imples that any Bargmann invariant may be realized by a circulant tuple of qutrit states.
    }
    \label{fig:1}
\end{figure}

\subsection{Summary of main results}
In this article, we provide a full characterization of the sets of Bargmann invariants, by proving the following results.
\begin{enumerate}
    \item As a main result, we prove that any Bargmann invariant is realizable using pure states associated with a circulant Gram matrix,
    \[
        B_n
        =
        \BargmannSet{n}
        =
        \BargmannSet{n\mid \mathrm{circ}}
        = \mathcal{P}_n^n.
        \]
    With slight abuse of notation, $\mathcal{P}_n^n$ is the set of numbers $z^n$ such that $z$ belongs to the unit $n$-gon $\mathcal{P}_n.$
    Our main contribution is the middle equality; the third equality was proven in \cite{Li25}, while the first follows as a corollary of our work.

    \item
    As a consequence, we show that the set $B_{n}$ is a teardrop-shaped convex region (see Fig.~\ref{fig:1}, panel a), with boundary points
    \begin{equation}
        z^n =
        -\left(
            \frac{\sec(\tau / n)}{\sec(\pi/n)}
        \right)^n
        e^{i\tau},
    \end{equation}
    for $-\pi \leq \tau < \pi.$

   The shape of this boundary was derived in Refs.~\cite{Li25, Zhang25}, and proven to be a boundary for circulant tuples of states in Ref.~\cite{Li25}. Here, we prove that it is also the border of the set of values that can be taken by general Bargmann invariants.

    \item Any realizable Bargmann invariant may be obtained as a Bargmann invariant of qubit states. Alternatively, it may be written using a tuple of states of dimension three (qutrit states) described by a circulant Gram matrix. In other words,
    \[
        \BargmannSet{n}
        =
        \BargmannSet{n,2}
        =
        \BargmannSet{n,3 \mid \mathrm{circ}}.
    \]
    
    Given a Bargmann invariant $\Delta = \abs{\Delta}e^{i\theta}$ of order $n$, its qubit realization is $\braket{x}{v_1} \braket{v_1}{v_2} \cdots \braket{v_{n-1}}{x}$, where
    \begin{equation}
        \ket{v_k}
            = \cos{\varphi} \ket{0}
            +  \sin{\varphi}\,\omega_n^k \ket 1,
        \quad
        1 \leq k < n,
    \end{equation}
    the angle $\varphi$ is given by the equation
    \(
    \cos{(2\varphi)} 
    =
    \tan{(\frac{\pi-\theta}{n})} / \tan{(\frac{\pi}{n})},
    \)
    and $\ket{x}$ is given implicitly as in Theorem~\ref{thm:qubit_realization}.
    
    The qutrit realization is given by a tuple of states of the form
    \begin{equation}
        \ket{u_k}
        =
         a \ket{0}
        +  b \, \omega_n^k \ket{1}
        +  c \, y_n^k \ket {2},
        \quad
        0\leq k < n,
    \end{equation}
    where
    $y_n=\omega_n^{\ceil{n/2}}$ and $a,b,c$ are real numbers such that $a^2+ b^2 + c^2 = 1$ (see Corollary~\ref{corollary:qutrit_realization} for details).

\end{enumerate}
\section
{Bargmann Invariants of Circulant tuples of States}
We begin by characterizing the sets of circulant Bargmann invariants of pure states.
Our main contribution will be presented in the next section, where we show that the set of general Bargmann invariants coincides with the set of circulant Bargmann invariants.

The following theorem was derived in Ref.~\cite{Li25} in a slightly different form (Theorem 1). We provide a shorter a proof here, and point out the realizability of the Bargmann roots using qubit and qutrit states. The theorem is illustrated in Fig.~\ref{fig:1}, panel b).
\begin{theorem}
    The set of circulant Bargmann roots of order $n$ coincides with the unit $n$-gon,
    $\RootSet{n \mid \mathrm{circ}}
    = \mathcal{P}_n$ .
    Furthermore, its vertices, edges and interior points are realizable as circulant Bargmann roots of states with dimensions 1, 2 and 3, respectively.
    \label{thm:circulant_bargmann_roots_convexity}
\end{theorem}
\begin{proof}
    By definition, a number $z \in \RootSet{n \mid \mathrm{circ}}$ is a circulant Bargmann root of order $n$ if and only if there exists an $n \times n$ circulant Gram matrix $G$ such that $z = G_{01}$ and $G_{ii} = 1$ for $0 \leq i < n$.
    
    Now, $G$ is a circulant Gram matrix if and only if it is diagonalizable as $G = F D F^\dagger$,
    where $F$ is the unitary Fourier matrix ($F_{ij} = \omega_n^{ij} /  \sqrt{n}$, with $w_n = e^{i2\pi / n}$),
    and $D=\mathrm{diag}(\lambda_0, \ldots, \lambda_{n-1})$ is the matrix of eigenvalues $\lambda_0,\ldots,\lambda_{n-1} \geq 0$ (since $G$ must be positive semi-definite)~\cite[2.2.P10]{HornJohnson}. Note also that any choice of $\lambda_i\geq 0$ such that $\sum_i \lambda_i=n$ leads to a valid Gram matrix since $G_{ii}= (F D F^\dagger)_{ii}= \sum_i \lambda_i/n=1. $
    
    In short, $z$ is a circulant Bargmann root of order $n$ if and only if there exists a Gram matrix $G$ such that
    \begin{equation}
    z
    = G_{01}
    = 
    \sum_{i=0}^{n-1}
    \frac{\lambda_i}{n}
    \, \omega_n^{-i},
    \label{eq:circulant_bargmann_root_expression}
    \end{equation}
   where the coefficients $\lambda_i / n$ are non-negative, and their sum is $\Tr{G} / n = 1$.
    Therefore, for any $z$ there is a convex combination of the $n$-th roots of unity $\omega_n^i$, and vice-versa.
    We conclude that $z$ belongs to $\RootSet{n \mid \mathrm{circ}}$ if and only if $z$ belongs to $\mathcal{P}_n$, as wanted.

    From Eq.~\ref{eq:circulant_bargmann_root_expression}, we see that the vertices $\omega_n^{-k}$ correspond to taking the eigenvalues $\lambda_i/n=\delta_{ik},$ which results in rank-1 Gram matrices, and are thus realizable by one-dimensional states.
    The points on the boundary are convex combinations of vertices, so they must correspond to Gram matrices of rank 2 at most, and are thus realizable using two-dimensional states.

    Finally, Carathéodory's theorem guarantees that any point $z$ within $\mathcal{P}_n$ can be written as a convex combination of at most three vertices of $\mathcal{P}_n$. Following the reasoning of the previous paragraph, this implies that $z$ must be realizable as a Bargmann root of three-dimensional states.
\end{proof}
We now translate the theorem to circulant Bargmann invariants, which is the original form of Theorem 1 in Ref.~\cite{Li25}, except for the statement on realizability as qubits and qutrits.
\begin{corollary}
    The set of circulant Bargmann invariants of order $n$, $\BargmannSet{n \mid \mathrm{circ}}$, coincides with $\mathcal{P}_n^n$. Its boundary corresponds to circulant Bargmann invariants of states with dimension 2, and its interior to Bargmann invariants of states with dimension 3.
\end{corollary}
\begin{proof}
    This is a straightforward consequence of the fact that $\BargmannSet{n\mid\mathrm{circ}}$ is obtained from $\RootSet{n\mid\mathrm{circ}}$ by applying the map $f_n:z\to z^n$. Since $f_n$ is continuous, boundary points of $\BargmannSet{n\mid\mathrm{circ}}$ must be images of boundary points of $\RootSet{n\mid\mathrm{circ}}$, which, according to the previous theorem, 
    are realizable by states of dimension 2.
    Any interior point of $\BargmannSet{n\mid\mathrm{circ}}$ is obtained from some point of $\RootSet{n\mid\mathrm{circ}}$, so, according to the previous theorem, it is realizable by states of dimension 3 or lower; if lower, we may embed the states in a three-dimensional space, thus proving our statement.
\end{proof}

We now show that the sets $\BargmannSet{n\mid\mathrm{circ}}$ are convex and are teardrop-shaped (Fig.~\ref{fig:1}, panel a), using a variation of an argument found in Ref.~\cite[Theorem 2]
{Zhang25}.
The shape of the boundary was described in Refs.~\cite{Li25, Zhang25}.
\begin{corollary}
    The set $\BargmannSet{n\mid\mathrm{circ}}$ is convex.
    For $n \geq 3$, it is the convex hull of the boundary points
    \begin{equation}
        z^n
        =
        -\left(
            \frac{\sec(\tau / n)}{\sec(\pi/n)}
        \right)^n
        e^{i\tau},
    \end{equation}
    where $-\pi \leq \tau \leq \pi$.
    Otherwise, $\BargmannSet{1} = \{1\}$, and $\BargmannSet{2} = [0, 1]$.
    \label{cor:circulant_bargmann_invariant_convexity}
\end{corollary}
\begin{proof}
    The cases $n=1$ and $n=2$ are easy to verify.

    To show that $\BargmannSet{n\mid\mathrm{circ}} = \mathcal{P}_n^n$ is convex for $n \geq 3$, it suffices to show that the boundary has curvature of constant sign.
    Let $z(t)=1+t \delta$ be the points of the boundary of the first sector of $\mathcal{P}_n$, where $0 \leq t \leq 1$ and $\delta = e^{i\,2\pi/n} - 1$. It is enough to consider this segment, since all the others will be mapped to the same set under $f_n: z \to z^n$.

    Let $x(t) = z(t)^n$. For the curvature of a planar curve to change sign, it must be zero somewhere, which happens when the velocity $x'(t)$ is parallel to the acceleration $x''(t)$. Equivalenty, this happens when the ratio $x' / x''$ is real,
    \begin{equation}
        \frac{x'(t)}{x''(t)}
        =
        \frac{
            n \delta z^{n-1}
        }{
            n(n-1) \, \delta^2 \, z^{n-2}
        \
        }
        =
        \frac{t + \delta^{-1}}{n-1}.
    \end{equation}
    But $\delta$ is not real for $n \geq 3$, so the curvature cannot become zero. And since $x(t)$ is a simple curve---non-intersecting and closed,  as $x(0)=x(1)=1$---it must therefore bound a convex region. 

    Finally, we can parameterize the points $z(t)$ in a few different ways.
    First, we introduce the new parameter $\lambda=2t - 1$, such that
    \[
    z =
    e^{i\pi/n}
    \big(
        \cos{(\pi / n)}
        + i \lambda \sin{(\pi / n)}
    \big),
    \quad
    -1 \leq \lambda \leq 1.
    \]
    This parameterization captures the symmetry of the segment $z(t)$, since points corresponding to $\pm \lambda$ are symmetric about its midpoint ($\lambda = 0$).
    The points on the segment also admit a compact trigonometric parameterization.
    If we define a new parameter $\tau$ such that $\lambda = \tan(\tau/n) / \tan(\pi/n)$, then
    \begin{equation}
        z =
        \frac{
            \sec{(\tau / n)} \,
        }{
            \sec{(\pi/ n)} \,
        }
        e^{i(\pi + \tau) / n},
        \quad
        -\pi \leq \tau \leq \pi.
    \end{equation}
    Applying $f_n$, and using $e^{i\pi} = -1$, proves our final claim,
    \begin{equation}
        z^n =
        - \left(\frac{
            \sec{(\tau / n)}
        }{
            \sec{(\pi/ n)}
        }\right)^n
        e^{i\tau},
        \quad
        -\pi \leq \tau \leq \pi.
    \end{equation}
    We note that, if we define $\theta = \tau + \pi$, we obtain the parameterization
    \begin{equation}
        z^n
        =
        \cos^n
        \left( \frac{\pi}{n} \right)
        \,
        \sec^n
        \left(\frac{\theta - \pi}{n}\right)
        e^{i\theta},
        \quad
        0\leq \theta \leq 2\pi,
    \end{equation}
    which is the form obtained in Ref.~\cite{Zhang25}, Theorem 1, proven to be the boundary for general Bargmann invariants of pure states for $n=3, 4$ and conjectured for $n\geq 5$.
\end{proof}
To conclude this section, we provide an explicit realization of circulant Bargmann invariants using qutrit and qubit states.
The following construction is illustrated in Fig.~\ref{fig:1}, panel $b$.

\begin{corollary}
    Any circulant Bargmann invariant $z \in \BargmannSet{n\mid\mathrm{circ}}$  may be realized as a circulant invariant of the three-dimensional states
    \begin{equation}
        \ket{u_k}
        =
         a \ket{0}
        +  b \, \omega_n^k \ket{1}
        +  c \, y_n^k \ket {2},
    \end{equation}
    for $0\leq k < n$, where
    $y_n=\omega_n^{\ceil{n/2}}$, and $
    a^2, b^2, c^2$ are the convex coefficients of the principal root $z^{1/n}$ such that $a^2+ b^2 + c^2 = 1$ and
    \begin{equation}
        z^{1/n} = a^2
            + b^2\,\omega_n
            + c^2 \, y_n.
    \label{eq:qutrit_convex_combination}
    \end{equation}
    The boundary of $\BargmannSet{n\mid\mathrm{circ}}$ corresponds to circulant Bargmann invariants of two-dimensional states, which are realizable by taking $c=0$, resulting in the Oszmaniec-Brod-Galvão qubit states~\cite{oszmaniec_measuring_2024}
    \begin{equation}
        \ket{v_k(\varphi)}
        =
        \cos{\varphi} \ket{0}
        + \sin{\varphi} \, \omega_n^k \ket{1},
        \quad
        0\leq k < n,
        \label{eq:OBG_vectors}
    \end{equation}
    for $0 \leq \varphi \leq \pi/2$. 
    \label{corollary:qutrit_realization}
\end{corollary}
\begin{proof}
    Let $z \in \BargmannSet{n\mid\mathrm{circ}}$.
    Its principal root $z^{1/n}$ must lie in the first triangular sector of $\mathcal{P}_n$ --- it is the convex combination of the points $(0, 1, \omega_n)$. Now, let $y_n = \omega_n^{\ceil{n/2}}$. If $n$ is even, then $0 = 1 + y_n $; if $n$ is odd, then $0 =  K y_n + (1 + \omega_n) / 2$, with $K=\sqrt{1/2 + \cos(2\pi /n)/2}$. In either case, zero can be written as a linear combination of $(1, \omega_n, y_n)$. Therefore, $z^{1/n}$ can also be written as their convex combination, so the numbers $a,b,c$ of Eq.~\ref{eq:qutrit_convex_combination} exist.
    
    If $z$ lies on the boundary, it must lie on the segment joining 1 to $\omega_n$, so we may assume  $(a,b,c)=(\cos{\varphi}, \sin{\varphi}, 0)$ with $0 \leq \varphi \leq \pi/2$.
    We may calculate $\varphi$ as follows, given $z = \abs{z} \, e^{i \theta}$, with $0\leq \theta \leq 2\pi$. Using the parameterizations of the proof of Corollary~\ref{cor:circulant_bargmann_invariant_convexity} we see that $\sin^2\varphi = t = (1 + \lambda) / 2$, and that $\theta$ matches the $\theta = \tau + \pi$ of that Corollary. The equality $\lambda = \tan(\tau / n) / \tan(\pi / n)$ allows us to complete this chain of definitions. Using the identity $\sin^2 \varphi = (1 - \cos(2\varphi))/ 2$, we conclude that
    \begin{equation}
        \cos(2\varphi)
        =
        -\frac{\tan(\tau / n)}{\tan(\pi/n)}
        =
        \frac{\tan(\frac{\pi - \theta}{n})}{\tan(\pi/n)}.
    \end{equation}

    It is straightforward to verify that the vectors $\ket{u_k}$ give rise to a circulant Gram matrix and that
    \begin{equation}
        \braket{u_0}{u_1}
        \cdots
        \braket{u_{n-1}}{u_0}
        =z^{1/n} \cdots z^{1/n}
        = z.
    \end{equation}
\end{proof}

\section
{General Bargmann Invariants}

In this section we prove the main result of our work: any Bargmann invariant may be written as a circulant Bargmann invariant.
We also prove that any Bargmann invariant is realizable using qubit states only.

We begin by establishing the result for Bargmann roots, where the geometric intuition is more natural.
We will do this by showing that any Bargmann root is, in a sense, interior to some circulant Bargmann root.
The construction in the next proof is illustrated in Fig.~\ref{fig:2}.

\begin{figure}
    \centering
    \includegraphics[width=.6\linewidth]{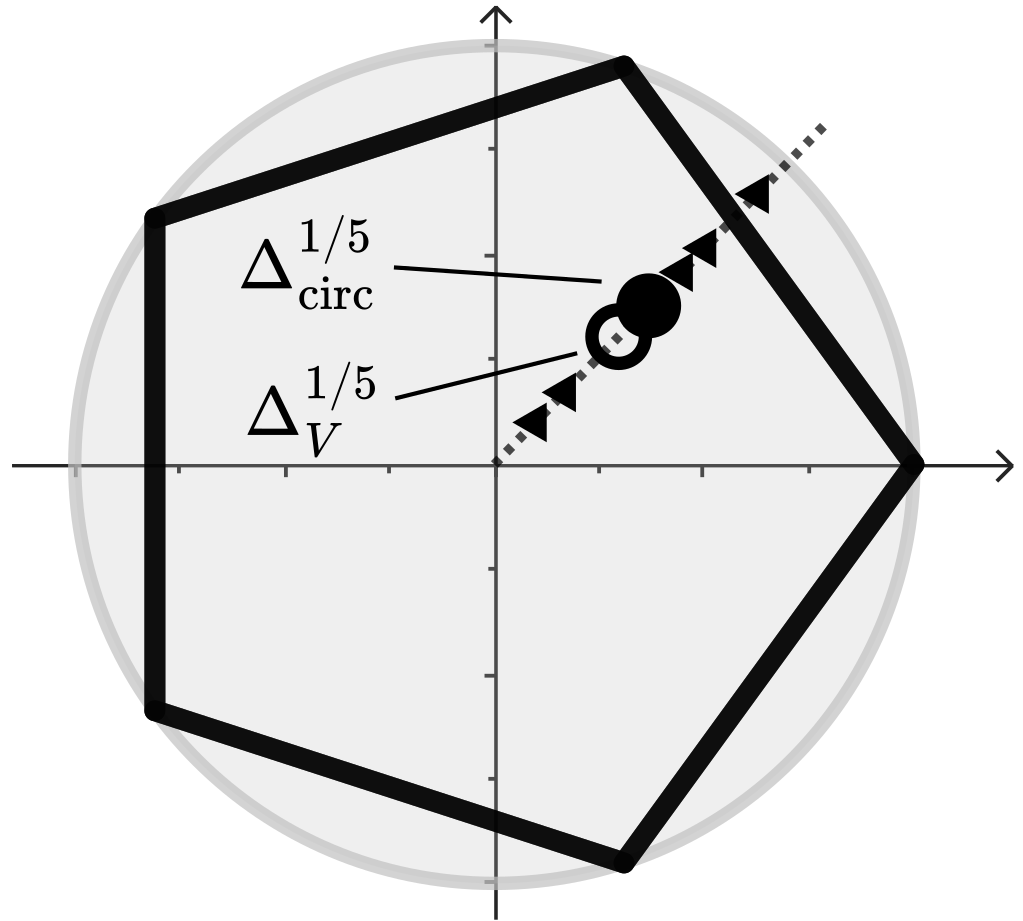}
    \caption{
    Illustration of Theorem~\ref{thm:main_theorem} for $n=5$.
    Given a tuple of states $V=(\ket{v_1},\ldots,\ket{v_n})$ we can use gauge freedom to map the inner products $\braket{v_1}{v_2}, \ldots, \braket{v_n}{v_1}$ (black triangles) onto the same complex ray, without changing the corresponding Bargmann invariant $\Delta_V$. The geometric mean of the inner products (open circle), which corresponds to the $n$-the root of their Bargmann invariant, must lie closer to the origin than their arithmetic mean, $\Delta^{1/n}_\text{circ}$ (filled circle). Both the origin and $\Delta^{1/n}_\text{circ}$ are roots of Bargmann invariants that are realizable by circulant tuples of states. Since the set $\RootSet{n\mid\mathrm{circ}}$ is convex, $\Delta_V$ must be realizable as a Bargmann invariant of a circulant tuple of states.
    }
    \label{fig:2}
\end{figure}

\begin{theorem}
     The set of Bargmann roots $\RootSet{n}$ coincides with $\RootSet{n \mid \mathrm{circ}}.$
     In particular, any Bargmann root is realizable as a circulant Bargmann root.
     \label{thm:main_theorem}
\end{theorem}
\begin{proof}
    By definition, $\RootSet{n \mid \mathrm{circ}} \subseteq \RootSet{n}$. Let us prove the converse, that is, $\RootSet{n} \subseteq \RootSet{n \mid \mathrm{circ}}$.

    Let $z \in \RootSet{n}$ be a Bargmann root of order $n$. This means that there exist vectors $\ket{v_1}, \ldots, \ket{v_n} \in \mathcal{H}$,
    for some Hilbert space $\mathcal H$, such that
    \begin{equation}
        z^n = \braket{v_1}{v_2} \cdots \braket{v_n}{v_1}.
    \end{equation}
    
    Now, let us define new vectors $\ket{u_i} = e^{i \phi_i} \ket{v_i}$, for some phases $\phi_1, \ldots, \phi_n$, such that
    \begin{equation}
        \arg{\braket{u_1}{u_2}}
        = \cdots
        = \arg{\braket{u_n}{u_1}}
        = \theta = \arg{z}.
    \end{equation}
    Clearly, we still have
    \(
        z^n = \braket{u_1}{u_2} \cdots \braket{u_n}{u_1}.
    \)
    A simple solution is to define $\phi_1 = 0$ and, recursively, $\phi_i =\phi_{i-1} + \theta - \arg{\braket{v_{i-1}}{v_i}}$.

    Now, let us apply the inequality of arithmetic and geometric means (AM-GM inequality) to conclude that
    \begin{equation}
        \sqrt[n]{
            \abs{\braket{u_1}{u_2}}
            \cdots
            \abs{\braket{u_n}{u_1}}
            }
        \leq
        \frac{
            \abs{\braket{u_1}{u_2}}
            + \cdots
            + \abs{\braket{u_n}{u_1}}
        }{n}.
    \end{equation}
    Let us re-write this inequality by defining
    \begin{align}
        g &= \frac{
            \braket{u_1}{u_2}
            + \cdots
            + \braket{u_n}{u_1}
        }{n} \\
         &= \frac{
            \abs{\braket{u_1}{u_2}}
            + \cdots
            + \abs{\braket{u_n}{u_1}}
        }{n}
        \,
        e^{i\theta},
    \end{align}
    where we used the fact that the arguments of the inner products $\braket{u_1}{u_2},\ldots,\braket{u_n}{u_1}$ are all equal to $\theta$.
    The AM-GM inequality can then be written as
    $ \abs{z} \leq \abs{g}$.
    Since $z$ and $g$ are complex numbers with the same argument, they must lie on the same ray, with $z$ closer to the origin.
    In other words, $z$ lies on the segment joining zero and $g$. We can see this formally by defining the ratio $\lambda = \abs{z} / \abs{g} \in [0, 1]$. Using $z = \abs{z} e^{i\theta}$ and $g = \abs{g} e^{i\theta}$,
    \begin{align}
        z
        = \lambda g
        = \lambda g + (1-\lambda)\cdot 0,
    \end{align}
    which, by definition, means that $z \in \mathrm{conv}\{0, g\}.$
    
    Finally, we will show that there is a normalized circulant matrix $G_\mathrm{circ}$ with $[G_\mathrm{circ}]_{01} = g$.
    Let $G$ be the Gram matrix of the vectors $\ket{u_0}, \ket{u_1}, \ldots, \ket{u_{n-1}}$, and define the matrices
    \begin{equation}
        G_k = (C_n^k) \, G \, (C_n^k)^\dagger,
    \end{equation}
    which are Gram matrices arising from cyclic permutations of $\ket{u_0}, \ket{u_1}, \ldots, \ket{u_{n-1}}$.
    Their arithmetic mean,
    \begin{equation}
        G_\mathrm{circ} = \frac{G_0 + \cdots + G_{n-1}}{n},
    \end{equation}
    is also a normalized circulant Gram matrix, because it is the convex combination of Gram matrices, and $[G_\mathrm{circ}]_{01} = g,$ as wanted.
    
    Therefore, $z$ is in the segment joining $0$ and $g$, numbers that belong to $\RootSet{n \mid \mathrm{circ}}$. Since $\RootSet{n \mid \mathrm{circ}}$ is convex, $z$ must belong to it as well.
    We thus conclude that $\RootSet{n} = \RootSet{n \mid \mathrm{circ}}$, as wanted.
\end{proof}

Our main result now follows as a corollary.

\begin{corollary}
    The sets $B_n$ and $\BargmannSet{n}$ coincide with $\BargmannSet{n \mid \mathrm{circ}} = \mathcal{P}_n^n.$
    In particular, the sets are convex and any Bargmann invariant may be realized as a circulant Bargmann invariant of three-dimensional states (i.e., using qutrit states).
    \label{cor:Bn_equals_Bn_circ}
\end{corollary}
\begin{proof}
    First, the sets $\BargmannSet{n}$ and $\BargmannSet{n \mid \mathrm{circ}}$ are equal, since they are obtained from $\RootSet{n}$ and $\RootSet{n \mid \mathrm{circ}}$, respectively, by application of the map $f: z \to z^n$, and these last two sets are equal.

    Now, let $\rho_1, \ldots, \rho_n \in \mathcal D(\mathcal H^d)$ be density matrices of dimension $d$ with spectral decompositions
    \begin{equation}
        \rho_i
        = \sum_{j=0}^{d-1}
            \lambda^{(i)}_j \ketbra{v^{(i)}_j},
    \end{equation}
    where $\lambda^{(i)}_j \geq 0$ and $\lambda^{(i)}_0 + \cdots + \lambda^{(i)}_{d-1} = 1$ for all
    $1 \leq i \leq n$ and $0\leq j <d$.
    Therefore,
    \begin{align*}
        \Delta &= \Tr{\rho_1 \cdots \rho_n} \\
        &=
        \sum_{\alpha \in A}
        \Lambda_\alpha \Tr{
            \ketbra{v_{\alpha_0}^{(1)}}
            \cdots
            \ketbra{v_{\alpha_{n-1}}^{(n)}}},
    \end{align*}
with $A=\{0, \ldots, d-1\}^n$ and $\Lambda_\alpha = \lambda_{\alpha_0}^{(1)} \cdots \lambda_{\alpha_{n-1}}^{(n)} \geq 0$ for all $\alpha \in A.$
It is clear that $\sum_{\alpha \in A} \Lambda_\alpha = 1$, so $\Delta$ is a convex combination of Bargmann invariants of pure states and $B_n \subseteq \operatorname{conv} \BargmannSet{n} = \BargmannSet{n}$. On the other hand, by definition, $\BargmannSet{n} \subseteq B_n$, which proves that $B_n = \BargmannSet{n}$.
\end{proof}

Finally, we show that any Bargmann invariant may be realized using qubit states. Having established Corollary~\ref{cor:Bn_equals_Bn_circ}, this follows as a consequence of Theorem 2 in Ref.~\cite{Li25}.
Our proof shows how to construct an explicit realization.
\begin{theorem}
    The set $\BargmannSet{n, 2}$ coincides with $\BargmannSet{n}.$
    In particular, any Bargmann invariant is realizable as a qubit Bargmann invariant (but not necessarily of circulant form).
    \label{thm:qubit_realization}
\end{theorem}
\begin{proof}
    By definition, we have $\BargmannSet{n, 2} \subseteq \BargmannSet{n}$. Let us prove the converse.
    
    Let $\Delta \in \BargmannSet{n}$. Consider the ray starting at zero and passing through $\Delta$. Because $\BargmannSet{n}$ is compact, the ray must cross the boundary at a point $\Delta'$, which must be realizable as a circulant Bargmann invariant of two-dimensional states. In other words, there exists $0 \leq \varphi \leq \pi/2$ such that
    \begin{equation}
        \Delta'
        = \braket{v_0}{v_1} \cdots \braket{v_{n-1}}{v_0}
        = \mel{v_0}{V}{v_0},
    \end{equation}
    where the $\ket{v_k} = \ket{v_k(\varphi)}$ are defined in Eq.~\ref{eq:OBG_vectors}, and
    \begin{equation}
        V = \ketbra{v_1} \ketbra{v_2} \cdots \ketbra{v_{n-1}}.
    \end{equation}
    Now, take $\ket{v_1^\perp} \in \mathcal{H}_2$ orthogonal to $\ket{v_1}$, so that
    \begin{equation}
        0
        = \braket{v_1^\perp}{v_1} \cdots \braket{v_n}{v_1^\perp}
        = \mel{v_1^\perp}{V}{v_1^\perp}.
    \end{equation}
    From this construction, we see that both 0 and $\Delta'$ belong to the \emph{numerical range} of $V$,
    \begin{equation}
        W(V)
        =
        \Big\{
        \mel{x}{V}{x}
        \, \big\vert \,
        \ket{x} \in \mathcal{H}_2, \,
        \braket{x} = 1
        \Big\}.
    \end{equation}
    According to the Hausdorff-Toeplitz theorem, the set $W(V)$ is convex, so it must contain $\Delta$ as well. In other words, there exists $\ket{x} \in \mathcal{H}_2$ such that
    \begin{equation}
        \Delta
        = \braket{x}{v_1} \braket{v_1}{v_2} \cdots \braket{v_{n-1}}{x},
    \end{equation}
    as wanted.
\end{proof}

\section{Conclusion}
In this work, we have provided a comprehensive and elementary characterization of the complex-valued range $B_n$ of Bargmann invariants, quantities crucial for the description of basis-independent coherence, imaginarity, Kirkwood-Dirac quasiprobabilities, OTOCs, and geometric phases.
We have shown that the range of all possible Bargmann invariants of order $n$ is precisely equal to the range of Bargmann invariants arising from tuples of pure qutrit states described by Gram matrices of circulant form. As the borders of these convex regions have been described in previous work, we now have a complete geometrical description of the range of each $B_n$.

Our characterization of $B_n$ gives a solution to many open problems that are related to the range of Bargmann invariants. It describes the possible extent of negativity and imaginarity in sequential weak values \cite{mitchison2007sequential} and extended Kirkwood-Dirac quasiprobability distributions \cite{halpern2018quasiprobability}. In this latter case, our results give an effective answer about the minimum invariant order required for a given value of negativity. Our results also bound the values that can be taken by higher-order OTOCs, proposed in \cite{abanin25} as tools for demonstrations of quantum computational advantage.

In addition, as pointed out in \cite{fernandes2024unitaryinvariant}, our results also find applications in the characterization of geometric phase experiments involving cyclic sequences of projective measurements, i.e. where a quantum state $\ket{\psi_1}$ undergoes a sequence of projections onto states $\ketbra{\psi_2}{\psi_2}$, $\ketbra{\psi_3}{\psi_3}, \ldots, \ketbra{\psi_n}{\psi_n}$, with a final projection on the initial state $\ketbra{\psi_1}{\psi_1}$. It is well known that the Bargmann invariant $\Delta_{12...n}$ associated to this tuple of states  determines both the post-selection probability,  given by $|\Delta_{12...n}|^2$, and the geometric phase acquired by the state which is given by $\text{arg}(\Delta_{12...n})$. Our results thus completely characterize the region of allowed values of the geometric phase and postselection probability of such experiments, whereas previous results only provided a characterization for experiments with 3 \cite{fernandes2024unitaryinvariant} and 4 \cite{Zhang25} projections. These geometric phases can be observed, for example, in a Mach-Zender interferometer, where one arm of the interferometer undergoes a cyclic series of $n$ projective measurements, implemented by polarization filters \cite{Berry09}.
Our bounds on both the absolute value and the phase of Bargmann invariants translate into precise trade-offs between the fraction of photons absorbed in such experiments, and the geometric phase that can be observed.

Our results also bear directly on the characterization of multiphoton indistinguishability. Linear-optical experiments featuring multiphoton inputs have output event probabilities that depend only on unitary-invariants of the spectral functions describing each photon's internal degrees of freedom. A paradigmatic two-photon example is the Hong-Ou-Mandel (HOM) effect  \cite{HOM}, where the probability of bunching at the output of a balanced beam-splitter gives a direct estimate of the overlap between the two single-photon spectral functions describing e.g. polarization and time-of-arrival. Experiments with 3 or more photons can depend on so-called collective photonic phases \cite{Tich14, shchesnovich2015partial, Menssen2017}, which are just the phases of Bargmann invariants of the spectral functions. Such experiments have been performed for 3 \cite{Menssen2017} and 4 \cite{Jones2020} photons; our characterization gives quantitative trade-offs between the level of two-photon indistinguishability, given by overlaps/HOM visibilities, and the measurable collective photonic phases of any order.

We also showed that any Bargmann invariant is realizable using qubit states only, which may find important applications in experiments. A simple consequence is that the value of a single Bargmann invariant cannot be used to witness Hilbert space dimension. This points to the interest in better understanding the information revealed by the values of multiple invariants of a single set of states. It is known that the values of multiple invariants \textit{can} be used to witness Hilbert space dimension \cite{giordani2021witnessing}, and in fact to solve any unitary-invariant problem associated with the set of states \cite{chien2016characterization, oszmaniec_measuring_2024}. Another application that could profit from a better understanding of the range of multiple Bargmann invariants is the study of multiphoton indistinguishability, where such results for multiple overlaps have already been used to simplify experimental set-ups \cite{giordani2020experimentalquantification}.

\textit{Note added:} While finishing this manuscript, we became aware of the related and independent preprint by J. Xu \cite{xu25}, which also provides a complete characterization of $B_n$, albeit without the explicit description of realizations using circulant qutrit states.

\begin{acknowledgments}
S.S.P.\ acknowledges FCT - Fundação para a Ciência e Tecnologia, I.P. in the framework of the projects UIDB/04564/2020 and UIDP/04564/2020, with DOI identifiers 10.54499/UIDB/04564/2020 and 10.54499/UIDP/04564/2020, respectively.
J.G.\ acknowledges support by Centro de Matem\'atica da Universidade de Coimbra (CMUC) - UID/MAT/00324.
E.F.G. and L.N. acknowledge support from FCT – Fundação para a Ciência e a Tecnologia (Portugal)
via project
CEECINST/00062/2018, and from Horizon Europe project EPIQUE (Grant No. 101135288).
\end{acknowledgments}

\bibliography{references}

\end{document}